\documentclass[12pt]{article}

\usepackage{amssymb,amsmath,amsthm,sectsty,url,graphicx}
\usepackage[letterpaper,hmargin=1.0in,vmargin=1.0in]{geometry}
\usepackage{color}
\usepackage[boxed]{algorithm}
\usepackage[margin=20pt,font=small,labelfont=bf]{caption}
\floatname{algorithm}{Figure}

\sectionfont{\large}
\subsectionfont{\normalsize}
\numberwithin{equation}{section}

\newtheorem{theorem}{Theorem}[section]
\newtheorem{lemma}[theorem]{Lemma}
\newtheorem{proposition}[theorem]{Proposition}
\newtheorem{corollary}[theorem]{Corollary}

\newtheorem{definition}[theorem]{Definition}

\renewcommand{\Pr}{ \mathrm P}

\newcommand{\F}{{\cal F}}

\newcommand{\N}{\mathbb N}

\newcommand{\Z}{\mathbb Z}

\begin{document}

\title{On giant components and treewidth in the layers model}
\author{Uriel Feige
\thanks{
Department of Computer Science and Applied Mathematics, Weizmann Institute of Science, Rehovot 76100,
Israel. E-mail: {\tt uriel.feige@weizmann.ac.il}. Work supported in part by The Israel Science Foundation (grant No. 621/12) and by the I-CORE Program of the Planning and Budgeting Committee and The Israel Science Foundation (grant No. 4/11).}
\and Jonathan Hermon
\thanks{
Department of Statistics, UC Berkeley. E-mail: {\tt jonathan.hermon@stat.berkeley.edu}.}
\and
Daniel Reichman
\thanks{
Department of Computer Science and Applied Mathematics, Weizmann Institute of Science, Rehovot 76100,
Israel. E-mail: {\tt daniel.reichman@gmail.com}. Work supported in part by The Israel Science Foundation (grant No. 621/12) and by the I-CORE Program of the Planning and Budgeting Committee and The Israel Science Foundation (grant No. 4/11).}
}
\maketitle

\begin{abstract}
Given an undirected $n$-vertex graph $G(V,E)$ and an integer $k$, let $T_k(G)$ denote the random vertex induced subgraph of $G$ generated by ordering $V$ according to a random permutation $\pi$ and including in $T_k(G)$ those vertices with at most $k-1$ of their neighbors preceding them in this order. The distribution of subgraphs sampled in this manner is called the \emph{layers model with parameter} $k$. The layers model has found applications in studying $\ell$-degenerate subgraphs, the design of algorithms for the maximum independent set problem, and in bootstrap percolation.

In the current work we expand the study of structural properties of the layers model.
 We prove that there are $3$-regular graphs $G$ for which with high probability $T_3(G)$ has a connected component of size $\Omega(n)$. Moreover, this connected component has treewidth $\Omega(n)$. This lower bound on the treewidth extends to many other random graph models. In contrast, $T_2(G)$ is known to be a forest (hence of treewidth~1), and we establish that if $G$ is of bounded degree then with high probability the largest connected component in $T_2(G)$ is of size $O(\log n)$. We also consider the infinite two-dimensional grid, for which we prove that the first four layers contain a unique infinite connected component with probability $1$.

\end{abstract}

{\small
}
\newpage

\section{Introduction}
Given a finite graph $G(V,E)$, a permutation $\pi$ over its vertices and an integer $k \ge 1$, let $L_k(G,\pi)$ denote the $k$th layer of $G$ according to $\pi$, defined as the set of those vertices of $G$ that have exactly $k-1$ of their neighbors preceding them in $\pi$.  The union of the first $k$ layers is denoted by $T_k(G,\pi): = \bigcup_{i=1}^k L_i(G,\pi)$. By a slight abuse of notation we refer to the subgraph induced on $T_k$ also by $T_k$, and omit $G,\pi$ when clear from the context. We shall be interested in the case that $G$ is given and $\pi$ is chosen uniformly at random over all permutations, in which case $L_k(G)$  ($T_k(G)$, respectively) refer to the random variable corresponding to the set of vertices in the $k$-th level (subgraph of $G$ induced on first $k$ levels, respectively). The random permutation model is equivalent to the following local sampling model: every vertex $v$ of $G$ selects independently at random an ``age" $X_v$ from the uniform distribution $U[0,1]$, and then $L_k(G)$ is the random variable specifying the set of those vertices that have exactly $k-1$ younger neighbors. When dealing with infinite graphs we shall only use the local sampling model.

The above procedure for sampling vertices from a graph has several useful properties. For every graph $G$ and every permutation $\pi$ the graph $T_k(G,\pi)$ is $k$-degenerate \cite{Reichman}. Namely, every subgraph of $T_k(\pi)$ has a vertex of degree at most $k-1$. In particular, $T_1(\pi)$ is an independent set, and $T_2(\pi)$ is a forest.
Moreover, the expected number of vertices in $T_k$ is exactly $\mathbb{E}[|T_k|] = \sum_{v\in V} \min[1,\frac{k}{d_v + 1}]$, where $d_v$ denotes the degree in $G$ of $v$.
A well known consequence of the properties listed above is that every graph $G(V,E)$ has an independent set of size at least $\sum_{v\in V} \frac{1}{d_v + 1}$~\cite{Wei}. For additional properties and applications of the random permutation model see Section~\ref{sec:relatedwork}.

In this work we study connectivity properties of $T_k$ for small values of $k$. One aspect that we consider is the likely size of the largest connected component in $T_k$. Another aspect considered is the typical {\em treewidth} of $T_k$. (We briefly remind the reader of the definition of treewidth. A {\em tree decomposition} of a graph $G(V,E)$ is a tree $T$ whose nodes are labeled by subsets of vertices from $V$ (called {\em bags}) with the following two properties: every edge $(u,v)\in E$ is in some bag, and for every vertex $v \in V$ the bags containing $v$ form a connected subtree of $T$. The width of the tree decomposition is one less than the cardinality of the largest bag, and the treewidth of $G$, denoted by $tw(G)$ is the smallest width for which $G$ has a tree decomposition. It is well known that forests have treewidth~1.) Our main results refer to infinite sequences of graphs, in which $n$ denotes the number of vertices in the underlying graph. The term $o(1)$ denotes a term that tends to~0 as $n$ tends to infinity.

As we have remarked above, $T_2(G)$ is necessarily a forest and thus has treewidth at most~1. It turns out that when the degree of $G$ is bounded by some absolute constant $d$ the largest component in $T_2(G)$ is logarithmic in size:

\begin{theorem}
\label{thm:T2}
There is a constant $b$ such that for every $d$ and $n$, and every $n$-vertex graph of maximum degree $d$, with high probability the size of the largest connected component in $L_2(G)$ does not exceed $2^{bd}\log n$.
\end{theorem}

There are graphs, the complete binary tree being one such example, for which $T_2$ is likely to have a connected component of size $\Omega(\log n)$. See Section~\ref{sec:bin} for more details.
For irregular graphs, it is unavoidable that the bounds in Theorem~\ref{thm:T2} depend on $d$. This can be seen by taking $G$ to be a collection of $\sqrt{n}$ disjoint stars, where each star has a central vertex of degree $\sqrt{n}$. With probability bounded away from $0$, at least one of the centers of the stars survives in $T_2(G)$, and then $T_2(G)$ has a connected component of size $\sqrt{n}$.

The next theorem shows that the properties of having small connected components and small treewidth, held by the first two layers, do not carry over to the first three layers.

\begin{theorem}
\label{thm:T3}
There is an infinite sequence of 3-regular graphs such that for some $\delta > 0$, with probability $1 - o(1)$ (over the choice of $\pi$) $T_3$ has a connected component of size at least $\delta n$. Moreover, with probability $1 - o(1)$ (over the choice of $\pi$) $T_3$ also has treewidth at least $\Omega (n)$.
\end{theorem}

Observe that for a graph of maximum degree~2 (composed of paths and cycles), for every $\pi$ we have that $T_3(\pi) = G$. Hence for these graphs showing the existence of a large connected component is easy (it suffices that $G$ itself has a large connected component). However, such graphs have treewidth bounded by~2, hence they cannot serve as examples showing that for some graphs $T_3$ is likely to have large treewidth.

In proving that the treewidth is large we shall use the following lemma (for a proof see Section \ref{sec:treewidth}) which connects between the treewidth of a random graph and the probability with which it has a giant component. Given a parameter $p \in (0,1)$ and a graph $G$ we refer to $G_p$ as the graph obtained by deleting independently every vertex with probability $1-p$ and keeping it otherwise with probability $p$.

\begin{lemma}
\label{lem:tw_main}
Let $G(V,E)$ be an $n$-vertex undirected graph and suppose there is $p \in (0,1)$ such that there is a connected component of size $\zeta n$ in $G_{p}$ with probability at least $1-c^n$ where $\zeta,c \in (0,1)$ are constants that may depend on $p$ but not on $n$. Then for every $q \in (p,1],$ $G_q$ has treewidth $\Omega(n)$ with probability $1-\exp(-\Omega(n))$.
\end{lemma}

Lemma~\ref{lem:tw_main} can be used in order to establish linear treewidth in a wide range of models for random graphs, and not only for the family of graphs referred to in Theorem~\ref{thm:T3}. See for example Theorem~\ref{thm:random_tw} in Section~\ref{sec:treewidth}.

We also initiate a study of $T_k$ on the two-dimensional infinite grid ${\Z}^2$.
\begin{theorem}
\label{thm:grid}
The first four layers of ${\Z}^2$ will have a unique infinite connected component with probability~1.
\end{theorem}

\subsection{Related work}
\label{sec:relatedwork}

The properties of the layers model were used in~\cite{FR} in designing algorithms for finding large independent sets in graphs. Let $\alpha(G)$ denote the size of the maximum independent set in $G$. Given $G$ and an integer $k$, one generates at random $T_k$ as in the random permutation model, and then applies an approximation algorithm to find a large independent set in $T_k$.
As $k$ grows, $\alpha(T_k)$ becomes a better approximation for $\alpha(G)$.
Observe that $\alpha(T_2)$ can be computed exactly in polynomial time, because $T_2$ is a forest. In~\cite{FR} an algorithm was presented for approximating $\alpha(T_3)$ within a ratio of~$\frac{7}{9}$. It is based on the observation that the expected number of edges in $T_k$ is at most $\frac{k-1}{2}\mathbb{E}[|T_k|]$, and hence the average degree in $T_3$ is not expected to exceed~2. A question that was left open in~\cite{FR} is whether $\alpha(T_3)$ can be approximated within ratios better than $\frac{7}{9}$, perhaps even arbitrarily close to~1. This would indeed hold if one could show that $T_3$ is likely to have small treewidth (sublinear in the number of vertices of $T_3$). Unfortunately, our Theorem~\ref{thm:T3} establishes that there are graphs for which $T_3$ is likely to have linear treewidth. We remark here that in the context of the work of~\cite{FR}, it is important that Theorem~\ref{thm:T3} addresses classes of graphs of minimum degree at least~3, because a preprocessing stage of the algorithms of~\cite{FR} eliminates all vertices of degree at most~2 from $G$ before employing the random permutation.

The random permutation approach has applications beyond those of finding independent sets. In~\cite{ABW,Reichman} it was observed that for every graph $G$ and permutation $\pi$,
$T_k(G,\pi)$ is a {\em contagious} set for $G$ with respect to {\em bootstrap percolation} with parameter $k$. Namely, if every vertex of $T_k(G,\pi)$ is initially activated, and thereafter in an iterative manner every vertex that has at least $k$ active neighbors becomes active as well, then all vertices of $G$ eventually become active. Using this approach, one can obtain upper bounds on the size of the smallest contagious set in $G$, in terms of the degree sequence of $G$.

The first two layers are related to several works that upper bound the number of queries in property testing, local computation and online algorithms~\cite{vardi,NO}. The general framework is as follows. Given an $n$-vertex graph $G(V,E)$ with maximal degree $d$ (which is a constant independent of $|V|$), every vertex $v$ is assigned independently a label distributed as a uniform $[0,1]$ random variable $l(v)$. Given $v$, let $C_{mon}(v)$ be the set of vertices reachable from $v$ by a monotone decreasing path of vertices. The complexity of the algorithms in~\cite{vardi} depends on the size of $C_{mon}(v)$. Extending ideas from~\cite{NO}, \cite{vardi} prove that with high probability \emph{for every vertex} $|C_{mon}(v)|=O(c(d)\log n)$, where $c(d)$ is a constant depending only on $d$. Theorem~\ref{thm:T2} can be deduced from this result of~\cite{vardi}, though we provide a self contained proof (based on techniques from~\cite{vardi}) of this theorem.
It is also proved in~\cite{vardi} that there are $n$-vertex graphs for which with probability greater than $\frac{1}{n}$ there exists a vertex $v$ such that $|C_{mon}(v)|>\frac{\log n}{\log \log n}$. Our Proposition~\ref{pro:binarytree} can be shown to imply an improved lower bound of $|C_{mon}(v)| \ge \Omega(\log n)$.

A connected component that contains a linear fraction of the vertices of a graph is often referred to as a {\em giant component}. Much of our work concerns the likelihood of having a giant component in $T_k$ for small values of $k$. There has been extensive work on the formation of giant components in random graph models (see for example~\cite{Fount,KSud,Janson,Molloy}), and we mention here two theorems that are most relevant to our work.

One theorem concerns the random configuration model in graphs which are allowed to have parallel edges and self loops. Let $\bar{d}$ be a sequence of $n$ nonnegative integers and let $d_i$ be the $i$th element of $\bar{d}$ (we assume $\sum d_i$ is even). In the configuration model $G^*(n,\bar{d})$ each vertex has $d_i$ half-edges and we combine the half edges of the pairs by choosing uniformly at random a matching of all half-edges. Given a multigraph sampled according to the configuration model, Molloy and Reed~\cite{Molloy} provide a criterion for the existence of a giant component. The exact statement of their results involves some technical conditions and parameters that are omitted here.

\begin{theorem}
\label{thm:MolloyReed}
Given a degree sequence $\bar{d}$ for an $n$ vertex graph, let $\lambda_i$ denote the fraction of vertices of degree $i$, and let $Q(\bar{d}) = \sum_{i \ge 1} \lambda_i i(i - 2)$. Let $G$ be a graph with degree sequence $\bar{d}$ selected randomly according to the configuration model. If $Q(\bar{d}) > 0$, then $G$ is likely to have a giant component, and moreover, the probability of not having a giant component is exponentially small in $n$ (with the base of the exponent depending on $Q(\bar{d})$). If $Q(\bar{d}) < 0$ then $G$ is unlikely to have a giant component.
\end{theorem}


Another theorem relevant to our work is an immediate consequence of results of Fountoulakis \cite{Fount} and Janson \cite{Janson}. For completeness, its proof is sketched in Section \ref{sec:treewidth}.

\begin{theorem}
\label{thm:randomregular}
Let $G(V,E)$ be a random $d$-regular graph on $n$ vertices where $d$ is a fixed constant. Let $G_p(V_p,E)$ be a random vertex induced subgraph of $G$ in which every vertex is selected into $V_p$ independently with probability $p$. For every $\epsilon > 0$ there is some $\delta > 0$ (that may depend on $d$ but not on $n$) such that the following holds. If $p \ge \frac{1 + \epsilon}{d-1}$ then there exists $c>0$ such that $G_p$ has a connected component of size at least $\delta n$ with probability at least $1 - e^{-c n}$, where probability is taken over the joint distribution of choice of $G$ and $G_p$.
\end{theorem}

It is interesting to compare our Theorem~\ref{thm:T2} with Theorem~\ref{thm:randomregular}. Given an $n$-vertex $d$-regular graph with $d \ge 4$, every vertex is likely to be in $T_2$ with probability $\frac{2}{d+1} > \frac{1}{d-1}$, but nevertheless, $T_2$ is unlikely to have a giant component (contrary to what Theorem~\ref{thm:randomregular} might suggest). This is a case were the local dependencies in the sampling procedure (a vertex is included in $T_2$ only if it is young relative to the random ages of its neighbors) affect the global properties of the resulting graph (existence of giant components).

There has been some interest in the treewidth of the supercritical Erd\H{o}s-R\'{e}nyi random graph $G(n,\frac{1+\epsilon}{n})$ in which every edge is included independently with probability $p=\frac{1+\epsilon}{n}$ \cite{Gao,Lee,Wang}. It is known that with high probability the giant component in $G(n,\frac{1+\epsilon}{n})$ has treewidth $\Omega(n)$ \cite{Lee}. This result was proved building on the work of Benjamini, Kozma and Wormald \cite{Kozma} showing that with high probability the giant component in $G(n,\frac{1+\epsilon}{n})$ contains a subgraph $G'$ of size $\Omega(n)$ which is an expander. Our techniques (i.e., Lemma \ref{lem:tw_main}) provide a different proof for the aforementioned result concerning $G(n,\frac{1+\epsilon}{n})$. Moreover, as already noted it implies that a multigraph sampled from the configuration model satisfying the Molloy-Reed criterion will have treewidth of $\Omega(n)$ with high probability. We are not aware of a previous proof of this fact.

The family of graphs $G_n$ considered in the proof of Theorem~\ref{thm:T3} is that of an $n$-cycle plus a random matching. Observe that as $G_n$ is $3$-regular, $T_3$ is likely to contain roughly $\frac{3n}{4}$ vertices, and we show that $T_3$ is likely to have linear treewidth. We note that for certain other choices of roughly $\frac{3n}{4}$ vertices from $G_n$, the resulting graph is a forest and hence has treewidth~1. See~\cite{Bau}.

Theorem~\ref{thm:grid} concerns the infinite grid. Site Percolation in which vertices of the grid are included independently with some fixed probability $p$ has been studied extensively. See for example~\cite{Grimmett}. In particular, it is known that if $p$ is smaller than roughly~$0.556$ then with probability 1, there will \emph{not} be an infinite connected component in the resulting subgraph ~\cite{Berg}. Our proof of~Theorem~\ref{thm:grid} is based in part on the existing machinery developed in percolation theory, in combination with structural properties of the first four layers of the infinite grid. It is currently an open question whether $T_3({\Z^2})$ contains an infinite component with probability $1$.

Additional results regarding the layers model
can be found in~\cite{Hermon}. These include extensions of our Theorems~\ref{thm:T2} and~\ref{thm:T3} to infinite graphs, and extension of our Theorem~\ref{thm:grid} to infinite grids of dimension $d > 2$. For finite graphs, the results of~\cite{Hermon} extend results of Theorem~\ref{thm:T3} to a wider class of random graphs.

\subsection{Preliminaries}

Given a graph $G(V,E)$, the connected component containing $v$ is denoted by $C(v)$. We shall sometime denote $|V|$ by $n$. The maximal degree in $G$ is denoted by $\Delta$. An independent set is a subset of vertices that does not span an edge. A forest is a graph with no cycle. Given $u,v \in V$ the distance between $u$ and $v$ denoted by $d(u,v)$ is the length (number of edges) of the shortest path connecting $u$ and $v$ (if $u$ and $v$ are not connected, the distance is defined to be $\infty$). For $A,B \subseteq V$, $d(A,B)$ is $\min(d(u,v)|u \in A,v \in B)$. For two vertices $u,v$, write $u \sim v$ if $d(u,v)=1$ and $N(u)$ is the set of all vertices adjacent to $u$. Consider a family of graphs $G_n$ over $n$ vertices and let $\widetilde{G_n}$ be a subgraph of $G_n$ that is created by some random process. Given a property $A$ of graphs, we say $\widetilde{G_n}$ has property $A$ with high probability ($w.h.p.$) if $\lim_{n \rightarrow \infty}\Pr[\widetilde{G_n} \in A]=1$. For a positive integer $\ell$, $[\ell]$ is the set $\{1,...,\ell\}$.

Let $T=(U,F)$ be a finite rooted tree with root $r$. For $u \in U$ we say that $w \in U$ is a \emph{descendent} of $u$ if the unique path connecting $r$ with $w$ in $v$ passes through $u$. The subtree rooted at $u$ (which we denote by $T_u$) consists of $u$ and all decedents of $u$.

Recall that $H$ is minor of $G$ if $H$ can be obtained from $G$ by deletion of vertices, deletion of edges, and contraction of edges. The following lemma regarding treewidth is well known (see~\cite{Kloks}):
\begin{lemma}
\label{lemma:minor}
If $H$ is a minor of $G$ then $tw(G) \geq tw(H)$.
\end{lemma}

We use standard concentration results regarding random variables. The following is referred to as Chernoff's inequality:
\begin{lemma}
  \label{lemma:Chernoff}
  Suppose that $X = \sum_{i=1}^mX_i$ where every $X_i$ is a $\{0,1\}$-random variable with $\Pr(X_i=1)=p$ and the $X_i$s are jointly independent. Then for arbitrary $\eta \in (0,1)$,
  \[
  \Pr(X<(1-\eta)pm) \le \exp(-pm\eta^2/2).
  \]
\end{lemma}

The following is referred to as Azuma's inequality:
\begin{lemma} \label{lem:Azuma}
Let $X_0,...,X_n$ be a martingale such that for every $1 \leq k <n$ it holds that $|X_k-X_{k-1}|\leq c_k$. Then for every nonnegative integer $t$ and real $B >0$
$$\Pr(|X_t-X_0|\geq B) \leq 2\exp\left(\frac{-B^2}{\sum_{i=1}^tc_i^2}\right).$$
\end{lemma}

Throughout, when considering the configuration model with degree sequence $\bar{d}$, we shall refer to the inequality $Q(\bar{d}) > 0$ (see Theorem~\ref{thm:MolloyReed}) as the \emph{Molloy-Reed criterion}.

\section{The first two layers}

In this section prove Theorem~\ref{thm:T2}. In fact, we show that it follows from a known result of~\cite{vardi}. We now explain this connection.

Given ages $X_v$ to vertices of a graph (as in the local sampling view of the random permutation model), let $C_2(v)$ denote the connected component of $v$ in $T_2(G)$. We say that a path $P = (v_1, \ldots, v_t)$ is \emph{monotonically decreasing} if $X_{v_i} > X_{v_{i+1}}$ for all $1 \le i \le t-1$. Let $C_{mon}(v)$ denote the set of all vertices reachable from $v$ via monotonically decreasing paths. We call $C_{mon}(v)$ the {\em monotone component} of $v$.

\begin{proposition}
\label{pro:domination}
$\max_v[|C_2(v)|] \le \max_v[|C_{mon}(v)|]$.
\end{proposition}

\begin{proof}
Every connected component of $T_2(G)$ is a tree. Moreover, orienting the edges of the connected component from highest label to lowest label gives a directed tree with the highest labeled vertex at the root. Observe that a vertex in $T_2(G)$ can have at most one neighboring vertex with age larger than the age of the vertex. If follows that if $v$ is the highest labeled vertex in its connected component, then $C_2(v) \subset C_{mon}(v)$. As every component has a highest labeled vertex, the proposition holds
\end{proof}

The following Theorem is from~\cite{vardi}.

\begin{theorem}
\label{thm:vardi}
Given random ages to vertices of an $n$-vertex graph of degree at most $d$ it holds that $\max_v[|C_{mon}(v)|] \le 2^{bd}\log n$, for some universal constant $b$.
\end{theorem}

Theorem~\ref{thm:vardi} together with Proposition~\ref{pro:domination} proves Theorem~\ref{thm:T2}.

For completeness, Section~\ref{sec:vardi} in the appendix contains a proof of Theorem~\ref{thm:vardi} (based on the proof given in~\cite{vardi}).

\subsection{The first two layers in a complete binary tree}
\label{sec:bin}

Here we show that there are bounded degree graphs for which the size of the largest component in the first two layers is $\Omega(\log n)$ with high probability.
We use the notation $BIN_n$ to denote the complete binary tree with $n - 1$ vertices (where $n$ is a power of~2). Hence $BIN_n$ has $\log n$ levels, where level~0 is the root, level $\log n - 1$ contains the leaves, and level $i$ contains $2^i$ vertices.

\begin{proposition}
\label{pro:wholetree}
Let $k$ be an arbitrary power of~2. Then the probability that $T_2(BIN_k) = BIN_k$ (namely, all of $BIN_k$ survives in $T_2(BIN_k)$) is at least $2^{-2k}$.
\end{proposition}

\begin{proof}
A sufficient condition for the event that $T_2(BIN_k) = BIN_k$ is that for every $0 \le i \le \log k - 1$ and for every vertex $v$ of level $i$, its random age $X_v$ is in the range $\frac{2^i}{k} < X_v \le \frac{2^{i + 1}}{k}$. This condition is satisfied with probability

$$\prod_{i=0}^{\log k - 1} \left(\frac{2^{i}}{k}\right)^{2^i} =  2^{\sum_{i=0}^{\log k - 1} (i - \log k) 2^{i}}$$
The sum in the exponent (written backwards) is precisely:

$$k(-\frac{1}{2} - 2\cdot\frac{1}{2^2} - 3\cdot\frac{1}{2^3} - \ldots - \log k \cdot \frac{1}{k}) > -2k$$

Hence the probability that $T_2(BIN_k) = BIN_k$ is at least $2^{-2k}$.
\end{proof}

\begin{proposition}
\label{pro:binarytree}
With high probability, $T_2(BIN_n)$ has a connected component of size at least $\frac{\log n}{10}$.
\end{proposition}

\begin{proof}
Fix $\ell$ to be the smallest power of~2 satisfying $\ell \ge 2 + \frac{\log n}{5}$. For a vertex $v$ at level $\log n - \log \ell$ of $BIN_n$, let $T_v$ denote the subtree of $BIN_n$ containing $v$ and its descendants. Observe that $T_v$ is isomorphic to $BIN_{\ell}$, that there are precisely $n/\ell$ such subtrees, and that they are all disjoint.
Let $Y_v$ be the random event that the vertices of $T_v$, except possibly for its root vertex $v$ (the root is excluded because it is connected to a parent node outside of $T_v$), is in $T_2(BIN_n)$. This event depends only on the random $X_u$ ages given to vertices in $T_v$, and Proposition~\ref{pro:wholetree} establishes that $Pr[Y_v] \ge 2^{-2\ell}$. As the $Y_v$ events are independent across the choices of vertex $v$, the probability that no $Y_v$ event holds is at most $(1 - 2^{-2\ell})^{n/\ell}$ which tends to~0 as $n$ grows (by our choice of $\ell$). Hence w.h.p.\ at least one event $Y_v$ holds, in which case $T_2(BIN_n)$ has a connected component of size at least $\ell/2 - 1 \ge \frac{\log n}{10}$.
\end{proof}

\section{Graphs for which $T_3$ has linear treewidth w.h.p.}
\label{sec:treewidth}
In this section we prove Theorem~\ref{thm:T3}.
We begin by proving Lemma \ref{lem:tw_main}.

\begin{proof}
A balanced vertex-separator in a graph $G(V,E)$ is a set $S\subseteq V$ such that every connected component in $G \setminus S$ is of size at most $2|V|/3$. It is well known (e.g., \cite{Kloks}) that if $tw(G) \le w$ then $G$ has a balanced separator $S$ such that $|S| \leq w+1$.
Our strategy is to show that $G_q$ has with high probability a subgraph $H$ such that every balanced separator in $H$ is of size $\Omega(n)$, which implies the required result by Lemma \ref{lemma:minor}.
Suppose the probability that $G_{p}$ does not have a connected component of size $\zeta n$, where $\zeta$ is a fixed positive constant depending only on $p$, is at most ${c}^n$ with $c<1$. $G_{p}$ can be exposed in two stages: first keep every vertex independently with probability $q>p$. In the remaining graph $G_q=(V_q,E_q)$, keep every vertex independently with probability $p/q$. Set $r=1-p/q$. Let $c'$ be an arbitrary number in $(c,1)$. Let $s \in (0,1/100]$ be a sufficiently small constant to be determined later. Suppose towards a contradiction that with probability at least ${c'}^n$, every subgraph $H=(U,F)$ of $G_q$ has a balanced-separator of size strictly smaller than $sn$. Repeat iteratively the following procedure: While $G_q$ has a connected component $C$ of size at least $\zeta n$ and assuming we are in the $i$th iteration, find a balanced-separator $S_i$ of the graph induced on $C$ of size at most $sn$ and delete all vertices in $S_i$ from $C$. Call a subset $W$ of the vertices of $G_q$ \emph{good} if $|W| \ge \frac{\zeta n}{4}$ and if $W$ is a union of connected components in the subgraph induced on $V_q \setminus \cup_{j \leq i}S_j$. By our assumptions on $s$ and the definition of a balanced separator, the maximum number of disjoint good sets increase
by at least one in every such iteration. It follows that after at most $\frac{4}{\zeta}$ iterations there will be no component of size larger than $\zeta n$. The total number of vertices deleted is at most $\frac{4}{\zeta}sn$. Choose $s$ to be sufficiently small such that ${c'}^nr^{\frac{4}{\zeta}sn}>{c}^n$. Then we get that the probability there is no component of size $\zeta n$ in $G_{p}$ is strictly larger than ${c}^n$. A contradiction. This proves that for  $q \in (p,1)$, $G_q$ contains a subgraph with treewidth $sn=\Omega(n)$ with probability $1-\exp(-\Omega(n))$, hence with high probability the treewidth of $G_{q}$ is $\Omega(n)$. Hence $G$ has treewidth $\Omega(n)$ as well. This concludes the the proof of the lemma.
\end{proof}

As a warm up and so as to introduce some of our techniques, before proving Theorem~\ref{thm:T3} that concerns 3-regular graphs, we shall prove a similar theorem for graphs of maximum degree~3, with the existence of degree~2 vertices making the proof easier compared to the 3-regular case. Our starting point for this short diversion is Theorem~\ref{thm:randomregular} whose proof is presented for completeness.

\begin{proof}
We prove the result for a random $d$-regular multigraph sampled according to the configuration model. Using standard contiguity results, this Theorem applies also to random (simple) $d$-regular graphs. Details are omitted.
Consider a (multi)-graph $G'_p$ that is generated according to the configuration model with the following degree sequence: vertices of $V_p$ have degree~$d$ and are referred to as the main vertices, whereas vertices of $V \setminus V_p$ are each broken into $d$ vertices of degree~1 that we refer to as auxiliary vertices. Consider the largest connected component $C$ in $G'_p$ and suppose that it has at least $d+2$ vertices. The auxiliary vertices in $C$ all have degree~1, whereas main vertices in $C$ are each connected to at most $d$ auxiliary vertices. Hence removing the auxiliary vertices from $C$ leaves a connected subcomponent $C'$ composed only of main vertices, and its size satisfied $|C'| \ge |C|/d$. This subcomponent $C'$ forms a connected component in $G_p$.

To analyze $|C|$ we check the Molloy-Reed criterion. Using standard concentration results we may assume that $G_p$ has $pn$ vertices of degree $d$ and $d(1 - p)n$ vertices of degree~1. The Molloy-Reed criterion requires analyzing the sign of the expression $pd(d-2) - d(1-p)$, which is positive if and only if $p > \frac{1}{d-1}$. The value of $\delta$ can be chosen such that for $p \ge \frac{1 + \epsilon}{d-1}$ the Molloy-Reed criterion implies that $G'_p$ has a connected component of size at least $\delta d n$ with probability at least $1 - e^{-c n}$. Hence $G_p$ has a connected component of size at least $\frac{\delta d n}{d} = \delta n$.
\end{proof}

\begin{corollary}
\label{cor:treewidth}
Let $G(V,E)$ be a random $d$-regular (multi)-graph on $n$ vertices selected according to the configuration model. Let $G_p(V_p,E)$ be a random vertex induced subgraph of $G$ in which every vertex is selected into $V_p$ independently with probability $p>\frac{1+\epsilon}{d-1}$, $\epsilon$ being some small constant. Then with high probability $G_p$ has treewidth $\Omega(n)$, where the $\Omega$ notation hides constants that depend on $d$ and on $\epsilon$, but not on $n$.
\end{corollary}

\begin{proof}
Immediate consequence of Lemma \ref{lem:tw_main} and Theorem \ref{thm:randomregular}.
\end{proof}

We now show how by a simple transformation we can turn any $d$-regular graph (multi)graph $G$ to a graph $\overline{G}$ such that the applying the site percolation  process on $G$ with $p=\frac{3}{d+1}$ is essentially stochastically dominated by taking $T_3$ on $\overline{G}$. Given a graph $G$, $\overline{G}$ is obtained from $G$ by replacing every edge $(u,v)$ by a path of length $3$, $u-x_{uv}-y_{uv}-v$ where we add two new vertices $x_{uv},y_{uv}$ for every edge $(u,v) \in E$.
Observe that in a graph $H=(U,F)$ along with a collection of its vertices
$v_1,...,v_s$ such that for every $1\leq i<j\leq s$, $d(v_i,v_j)>2$, the
events $A_i:=\{v_i \in T_3(H)\}$ are mutually independent as $A_i$ depends
only on the ages of $v_i$ and its neighbors. Hence the events $A_i$ for $i \in
[s]$ depend on ages of pairwise disjoint sets of vertices.

\begin{theorem}
\label{main} For a fixed $d>2$ and arbitrarily large $m$ there exist a graph $\widetilde{G}$ of maximal degree $d$ with $m$ vertices, such that with high probability $T_3(\widetilde{G})$ has treewidth $\Omega(m)$.
\end{theorem}
\begin{proof}
Take $G$ from Theorem \ref{thm:randomregular}: namely an $n$-vertex $d$-regular (multi)graph sampled from the configuration model and examine $\overline{G}$. The number of vertices in $\overline{G}$ is $m = \Theta(nd)$ and it has maximal degree $d$. Since vertices of degree $2$ remain with probability $1$ in $T_3$ (where $T_3$ refers to the first three layers in $\overline{G}$), $T_3$ is distributed as the graph obtained by \emph{independently} keeping each vertex in $G \cap \overline{G}$ (that is, all vertices of degree larger than $2$) with probability $p=\frac{3}{d+1}$ and keeping the rest of the vertices of $\overline{G}$ with probability 1. Hence $G_p$ is a minor of $T_3(\overline{G})$.
By Corollary~\ref{cor:treewidth}, $G_p$ has linear treewidth (and its number of vertices is $\Omega(n)$), implying that $T_3(\overline{G})$ has treewidth $\Omega(m)$ (where in both these last uses of the $\Omega$ notation it hides terms that depend on $d$).
\end{proof}

We proceed now to prove Theorem~\ref{thm:T3}, showing that there are $3$-regular graphs for which (with high probability) $T_3$ has a linear sized connected component and linear treewidth.

Let $G(V,E)$ be a 3-regular $n$-vertex (we assume $n$ is even) random graph with $E$ composed of two disjoint sets of edges, $C$ and  $M$. The vertex set of $G$ is $[0,n-1]$, $|C| = n$ and these edges form a cycle connecting all vertices in the standard cyclic order from~0 to~$n-1$ (without loss of generality we can label the vertices of the cycle with $0$ till $n-1$). $|M| = n/2$ and these edges form a random matching.

The following Theorem implies the first part of Theorem~\ref{thm:T3}.

\begin{theorem}
\label{thm:L3}
There is some fixed $\delta > 0$ independent of $n$, such that with probability $1 - e^{-\Omega(n)}$, the subgraph induced on $T_3(G(V,E))$ has a connected component of size at least $\delta n$. The probability is taken both over the random choice of $M$ and over the random permutation $\pi$.
\end{theorem}

Throughout the proof of Theorem~\ref{thm:L3} presented below we shall compute the expectations of certain random variables. We expose the vertices one vertex at a time, according to the standard cyclic order starting with the vertex labeled by 1. By Azuma's inequality and as $G$ has bounded degree, all random variables concerned are highly concentrated around their expectations, and hence we are justified in assuming that their realized value is equal to their expectation up to negligible additive terms that do not affect our proof.

To prove Theorem~\ref{thm:L3} we first choose $\pi$, and choose $M$ only afterwards. 
Fix a random permutation $\pi$ over the vertices.

\begin{proposition}
Consider an arbitrary ordering $\pi$ of $V$. For every matching $M$ of the vertices of $V$ it holds that with respect to $\pi$, $T_2(G(V,C)) \subset T_3(G(V,E))$.
\end{proposition}

We use $V_2$ to denote $T_2(G(V,C))$.  The randomness of $\pi$ easily implies that the expected size of $V_2$ is $2n/3$. The expected number of connected components in $T_2(G(V,C))$ is $n/3$, because every vertex not in $V_2$ contributes exactly one connected component (by cutting the cycle once).

The above establishes that the average size of a connected component in $T_2(G(V,C))$ is~2.
In our proof we shall analyze the distribution of sizes of connected components in $T_2(G(V,C))$. We first show that no connected component is too large. (This of course follows also from Theorem~\ref{thm:T2}, but can be proven much more easily in our case.)

\begin{proposition}
\label{pro:nolargecomponent}
Almost surely, no connected component in $T_2(G(V,C))$ contains more than $O(\log n)$ vertices.
\end{proposition}

\begin{proof}
Consider a set $S$ of $\ell$ consecutive vertices on the cycle. For $S$ to form a connected component in $T_2(G(V,C))$, it is required that none of its vertices is in $L_3(G(V,C))$. The probability for any individual vertex to belong to $L_3(G(V,C))$ is exactly~$1/3$. Any two vertices that are neither neighbors in $G(V,C)$ nor share a common neighbor in $G(V,C)$ are independent with respect to containment in $L_3(G(V,C))$. Hence $S$ contains a subset of at least $\ell/3$ independent vertices, and the probability that none of them is in $L_3(G(V,C))$ is at most $\left(\frac{2}{3}\right)^{\ell/3}$. As there are only $n$ ways of choosing the starting location of the set $S$, a union bound implies that almost surely no component contains more than $O(\log n)$ vertices.
\end{proof}

We now show that not too many of the connected components in $T_2(G(V,C))$ are very small.
For vertex $i$ and parameter $1 \le k < n$, let $p_k$ denote the probability (over choice of random permutation $\pi$) that $i-1 \not\in V_2$, $i+k \not\in V_2$, whereas $i, \ldots i+k-1$ are in $V_2$, where arithmetic is performed modulo $n$. Namely, $p_k$ is the probability that $i$ is a prefix of a segment of exactly $k$ consecutive vertices that belong to $V_2$. Observe that the probability $p_k$ does not depend on the choice of vertex $i$, by symmetry.

\begin{proposition}
\label{pro:lowdegrees}
For $p_k$ as defined above, $p_1 = \frac{2}{15}$ and $p_2 = \frac{1}{9}$.
\end{proposition}

\begin{proof}
To analyze $p_1$, consider five consecutive vertices $a,b,c,d,e$ on the cycle $C$, and compute the probability (over choice of $\pi$) that $c \in V_2$ whereas $b,d \not\in V_2$ (hence $c$ serves as $i$ in the definition of $p_1$). This event happens if and only if $\pi(a) < \pi(b) > \pi(c) < \pi(d) > \pi(e)$.  The permutation $\pi$ can be thought of as a bijection from $\{a,b,c,d,e\}$ to $\{1,2,3,4,5\}$. The permutations satisfying the event are $*5*4*$ (6 permutations), $45*3*$ (2 permutations), and their reverses $*4*5*$ and $*3*54$ (in the notation above $*$ serves as a ``don't care" symbol). Hence $p_1 = \frac{16}{120} = \frac{2}{15}$.

To analyze $p_2$, consider six consecutive vertices $a,b,c,d,e,f$ on the cycle $C$, and compute the probability (over choice of $\pi$) that $c,d \in V_2$ whereas $b,e \not\in V_2$. This event happens if and only if $\pi(a) < \pi(b) > \pi(c)$ and $\pi(d) < \pi(e) > \pi(f)$. Each of these two events has probability~$1/3$ and they are independent, hence $p_2 = \frac{1}{9}$.
\end{proof}

Hence in expectation, $T_2(G(V,C)$ has $2n/15$ components of size~1, $n/9$ components of size~2, and hence $n/3 - 2n/15 - n/9 = 4n/45$ components of size~3 or more. These larger components contain $2n/3 - 2n/15 - 2n/9 = 14n/45$ vertices, and hence their average size is $7/2$.

Construct now an auxiliary multigraph $H$ with two sets of vertices, $U_1$ and $U_2$. Every component in $T_2(G(V,C))$ serves as a vertex in $U_1$, of degree equal to its size.
The set $U_2$ consists of the $n/3$ vertices (in expectation) of $V \setminus V_2$, each of degree~$1$. Observe that the set of edges introduced by the random matching $M$ (which is part of the description of $G$) is distributed exactly like the set of edges introduced by the configuration model for generating random graphs with vertex set and degree sequence as described for $H$. This configuration model gives a random multigraph $H$. In the multigraph $H$, let $K$ be the connected component of largest size. We claim that $T_3(G)$ has a component of size at least $2|K|/3$. This can be seen as follows. Every vertex $v$ of degree~1 in $H$ that is part of $K$ must be connected in $K$ to some vertex $u$ that has degree more than~1 in $H$ (as otherwise $|K| = 2$, a case that can be dismissed as having exponentially small probability). Hence removing $v$ from $K$ does not disturb connectivity of those vertices remaining in $K$. For every vertex $v$ removed for this reason from $K$, the other endpoint $u$ of its matching edge was in $U_1$ (because only $U_1$ vertices can have degree more than~1). Moreover, within the component of $T_2(G(V,C))$ that corresponds to $u$, the endpoint of this matching edge hits a unique vertex of $V$. Hence if $K$ had $K_1$ vertices of degree~1 in $H$ (regardless of whether these vertices belong to $U_1$ or to $U_2$), then after removing them it still contains a set of at least $\max[K_1, 2(|K|-K_1)]$ vertices from $V$ that form a connected component in $T_3(G)$. This expression is minimized when $K_1 = 2|K|/3$, giving $2|K|/3$. ({\bf Remark:} $T_3(G)$ is likely to have components significantly larger than $2|K|/3$, because vertices in $T_3(G(V,E)) \setminus T_2(G(V,C))$ also contribute to the formation of a giant component. However, this aspect is not needed for our proof.)

It remains to analyze the probable size of the largest connected component in $H$. The Molloy-Reed criterion implies that $H$ is likely to have a giant component iff $\sum_{i > 0} \alpha_i d_i(d_i - 2) > 0$, where $\alpha_i$ is the fraction of vertices of degree $d_i$. ({\bf Remark:} the Molloy-Reed criterion is applicable to graphs in which the maximum degree is bounded by roughly $n^{1/4}$. The maximum degree in $H$ is smaller than the size of the maximum component in $T_2(G(V,C))$, which as shown in Proposition~\ref{pro:nolargecomponent} is at most $O(\log n)$.) To employ the Molloy-Reed criterion we need to know the degree sequence of $H$. Part of it is implied by Proposition~\ref{pro:lowdegrees}. Following that proposition we inferred that in addition to components of size~1 and~2, $T_2(G(V,C))$ has $4n/45$ components of average size $7/2$. Proposition~\ref{pro:equal} implies that the worst case for us is when there are $2n/45$ components of size~3 and $2n/45$ components of size~4, with no larger components.

\begin{proposition}
\label{pro:equal}
Consider two vertices of degree $d$ and $d' \ge d + 2$. Then the expression $\sum_{i > 1} \alpha_i d_i(d_i - 2) > 0$ decreases by replacing them by vertices of degrees $d + 1$ and $d' - 1$.
\end{proposition}

\begin{proof}
Initially the contribution of the two vertices is $d(d-2) + d'(d'-2)$. After replacement it is $(d+1)(d-1) + (d'-1)(d'-3)$, which is smaller by $2(d' - d - 1)$.
\end{proof}

In summary, we may assume that the degree sequence of $H$ is as follows. There are $5n/45$ vertices of degree~2, $2n/45$ vertices of degree~3, and $2n/45$ vertices of degree~4. As the total sum of degrees is $n$, there are $21n/45$ vertices of degree~1 (which indeed gives $2n/3$ vertices in $H$, which is the sum of number of connected components in $T_2(G(V,C))$ and vertices in $V \setminus V_2$). The Molloy-Reed criterion gives (the $1/30$ term below is the result of dividing the common term $n/45$ by the total number of vertices $2n/3$):

$$\sum_{i > 0} \alpha_i d_i(d_i - 2) \ge \frac{1}{30}(21 \cdot 1 \cdot (-1) + 5 \cdot 2 \cdot 0 + 2 \cdot 3 \cdot 1 + 2 \cdot 4 \cdot 2) = \frac{1}{30}>0.$$

Hence the Molloy-Reed criterion has a strictly positive value. The proof of Theorem~\ref{thm:L3} is now complete.
We note that with some extra work the ideas in the proof above
can be applied to other random graph models: see \cite{Hermon}.

To prove Theorem~\ref{thm:T3}, it remains to prove that with high probability $T_3(G)$ has treewidth $\Omega(n)$. Observe that in the proof of Theorem~\ref{thm:L3}, the degree sequence of $H$ depends only on the random permutation $\pi$, but not on the random matching $M$. Hence fixing the degree sequence of $H$ to be that used in the proof of Theorem~\ref{thm:L3} (which holds almost surely), Theorem~\ref{thm:T3} follows from the following theorem that shows that a graph sampled from the configuration model satisfying the Molloy-Reed criteria will have with high probability linear treewidth.

\begin{theorem}
\label{thm:random_tw}
Let $G$ be sampled from the configuration model $G^*(n,\bar{d})$ with maximum degree $O(\log n)$, and suppose that $\sum_{i \ge 1} \lambda_i i(i - 2)>0$ where the notation is as in Theorem \ref{thm:MolloyReed}. Then with high probability $G$ has treewidth $\Omega(n)$.
\end{theorem}

\begin{proof}
We apply similar ideas to those of \cite{Janson}, (see also \cite{Fount}) where the main observation is that that a random subgraph of $G$ is distributed according to configuration model (with the degree sequence obtained after deletions). Fix $p \in (0,1)$ and let $\widetilde{d}$ be the degree sequence obtained in $G_p$ with $\widetilde{n}$ the number of vertices of $G_p$. The $G_p$ is distributed according to $G^*(\widetilde{n},\widetilde{d})$.

When $p$ is sufficiently close to one (in fact it suffices that $p>\frac{\sum_{i \ge 1} \lambda_i \cdot i}{\sum_{i \ge 1} \lambda_i \cdot i \cdot (i-1)}$-see \cite{Janson} Theorems 3.5 and 3.9) we get using Azuma's inequality and the bounded degree assumption that with probability $1-e^{-\Omega(n)},$ $G^*(\widetilde{n},\widetilde{d})$ satisfies the Molloy-Reed criterion. Hence with probability  $1-e^{-\Omega(n)}$, $G_p$ has a component of size $\Omega(n)$. The theorem now follows from Lemma \ref{lem:tw_main}.
\end{proof}

\section{The two-dimensional grid}

In this section we prove Theorem~\ref{thm:grid}, that the first four layers of the two dimensional infinite grid ${\Z}^2$ will have a unique infinite connected component with probability~1.

\subsection{Proof overview}

We begin by explaining the ideas behind the proof that for $G={\Z}^2$, $T_4(G)$
has an infinite connected component (also referred to as an \emph{infinite
cluster})
with probability $1$. As mentioned in the introduction, our proof of~Theorem~\ref{thm:grid}
is based in
part on the existing machinery developed in percolation theory, in combination
with structural properties of the first four layers of the infinite grid.
We now describe those structural properties and explain how they can be combined
with the existing machinery to imply the assertion of Theorem~\ref{thm:grid}.

Standard results (see Lemmas \ref{lem: 0-1law1} and \ref{lem:zero_one}) imply
that it suffices to prove that $(0,0)$ belongs to an infinite cluster of
$T_4(G)$ with some positive probability $\Theta$. The graph ${\Z}_{*}^2$
is defined
to be the graph whose vertex set is that of ${\Z}^2$ and two distinct vertices
$(x_1,y_1)$ and $(x_2,y_2)$ are connected if $|x_1-x_2|\leq 1$ \emph{and}
$|y_1-y_2| \leq 1$. Observe that:

\begin{itemize}
\item The vertex $(0,0)$ does not belong to an infinite cluster in $T_4(G)
\cap
{\Z}^2$ iff $(0,0)$ is surrounded by a simple cycle $C$ in $L_5(G) \cap {\Z}_{*}^2$.
\end{itemize}

Define $$V_{even}=\{v \in \Z^2:v_1+v_2 \equiv 0 \bmod{2}\}$$
and
$$V_{odd}=\Z^2 \setminus V_{even}$$
Let ${\Z}_{even}^2$ be the graph whose vertex set is $V_{even}$ in which
two vertices $u$ and $v$ are adjacent to each other iff $\|u-v\|_{\infty}=1$
(${\Z}_{odd}^2$ is similarly defined). We further observe that:

\begin{itemize}
\item Every connected component (and hence also every cycle) of $L_5(G) \cap
{\Z}_{*}^2$ (as a subgraph of $\Z_*^2$) is entirely contained either in ${\Z}_{even}^2$
or in ${\Z}_{odd}^2$. This follows because $L_5(G)$ is an independent set
in ${\Z}^2$.
\item Both ${\Z}_{even}^2$ and ${\Z}_{odd}^2$ are isomorphic to ${\Z}^2$
(see Lemma \ref{lem: ZevenisomorphictoZ2}).
\end{itemize}

Our final observation is the essence of our proof that $T_4(\Z^2)$ contains
an infinite cluster $a.s.$ We argue that the percolation model restricted
to $V_{even}$ (respectively, $V_{odd}$) in which every vertex remains if
it belongs to $L_5(\Z^2)$ and is deleted otherwise is stochastically dominated
by the product measure with density $1/2$ on $V_{even}$ (respectively, $V_{odd}$)
denoted by $\Pr_{1/2}^{V_{even}} $ (respectively, $\Pr_{1/2}^{V_{odd}}$).
We provide a proof for $V_{even}$. We give a short review on stochastic domination
later in this section.
\begin{lemma}
\label{lem:order}
For every vertex $v \in {\Z}_{even}^2$ define a zero-one random variable
$Y_v$ which equals $1$ if $v$ belongs to $L_5({\Z}^2)$ and $0$ otherwise.
Then the law of the random field $(Y_v:v \in {\Z}_{even}^{2})$ is stochastically
dominated by the  product measure $\Pr_{1/2}^{V_{even}}$.
\end{lemma}
\begin{proof}
For $v \in V_{even}$ let $\widehat{v}:=v+(0,1)$. Define the indicator random
variable $Z_v$ to equal one if $X_v>X_{\widehat{v}}$ and zero otherwise (recall
that $X_v$ is the age of $v$). Clearly if $v$ is in $V_{even}$, then $\widehat{v}$
is in $V_{odd}$. In fact, we get a bijection from $V_{even}$ to $V_{odd}$.
Observe that $Z_v=0$ implies deterministically that $v \notin L_5$. The lemma
follows as the random variables $(Z_v : v \in V_{even})$ are jointly independent
and as $X_v$ and $X_{\widehat{v}}$ are \emph{independent}, $\Pr(X_v>X_{\widehat{v}})=\Pr(X_v<X_{\widehat{v}})=\frac{1}{2}$,
as required.
\end{proof}

The above sequence of observations reduces the question of the existence
of an infinite component
in $T_4(G)$ to that of the non-existence of a cycle around $(0,0)$ in
(independent) site percolation on ${\Z}^2$ with $p=1/2$. It is known that
site percolation on ${\Z}^2$ with $p=1/2$ is subcritical
(the probability that an infinite cluster exists is zero), and moreover,
that the probability that a vertex belongs to a component of diameter $\ell$
decays exponentially in $\ell$. This implies that the probability that there
is a cycle around $(0,0)$ in ${\Z}_{even}^2
\cap
L_5$ of length greater than $\ell$ tends to zero as $\ell$ tends to infinity,
and likewise for ${\Z}_{odd}^2
\cap L_5$.
Hence for large enough $\ell$ with positive probability there are no such
cycles. Standard arguments from percolation theory then
imply that with positive probability $(0,0)$ belongs to an infinite cluster
in $T_4$.

We also observe that similarly to the case of independent site
percolation on $\Z^2$, the aforementioned exponential decay implies that
our results for infinite grids can be scaled to \emph{finite} boxes in ${\Z}^2$. See Theorem~\ref{thm: scalling} for more details.

\subsection{Percolation Background}
We give a few definitions and lemmas from percolation theory, following
\cite{Grimmett}. For a countable set $S$ and $p \in[0,1]$, let $\Pr_p^{S}$
be the product
measure on $\Omega:=\{0,1\}^{S}$ with density $p$. That is, for every $s
\in S$ let $open(s):=\{w\in \Omega:
w(s)=1  \}$, then for any $s \in S$, $\Pr_p^{S}[open(s)]=p$ and the events
$(open(s) : s \in S )$ are independent with respect to $\Pr_p^S$. When $S$
is the vertex set of a graph $G$, $\Pr_p^{S}$ is simply the measure corresponding
to \emph{independent site percolation} on $G$ to be defined shortly. Equip
$\Omega$ with the cylinder $\sigma$-algebra
$\F$ (the $\sigma$-algebra generated by the events $(open(s) : s\in S)$)
and with the partial order
$\leqslant$, where $w \leqslant w'$ if $w(s) \le
w'(s)$ for all $s \in S$. We say that an event $A$ is increasing if $w \in
A$ and $w \leqslant w'$ implies that also $w' \in A$. For any two probability
measures $\mu$ and $\nu$ on $(\Omega,\F)$, we say that $\mu$ stochastically
dominates $\nu$ if $\mu(A) \ge \nu(A)$ for any increasing event $A$. It is
well-known and easy to show that if $(X_s)_{s \in S}$ and $(Y_s)_{s \in S}$
are random variables defined on the same probability space $(\Omega,\F,\Pr)$ such
that $\Pr[ X_s \ge Y_s]=1$ for all $s \in S$, then the law of $(X_s)_{s \in S}$ stochastically dominates that
of $(Y_s)_{s \in S}$. In simpler words, if two distributions on the space $\{0,1 \}^S$ can be coupled such that the first is point-wise larger than the other with probability one, then the first stochastically dominates the other. This was used in the proof of Lemma \ref{lem:order}.

 Recall that in \emph{site percolation}, every vertex $v$ of $G=(V,E)$ is
associated with a $0$-$1$ valued random variable $Y_v$. Formally, in the
above
notation the percolation process is defined on a probability space $\{\Omega,\F,\Pr
\}$ with $S=V$ and $Y_v=1_{open(v)}$. The most widely studied case is that
of \emph{independent (Bernoulli) percolation} where $\Pr=\Pr_p:=\Pr_p^V$
for some $p \in [0,1]$. However the definitions apply also when there may
exist
dependencies. A surviving vertex $v$ (i.e.\ $Y_v=1$) is called \emph{open}
and a deleted vertex $v$ (i.e.\ $Y_v=0$) is called \emph{closed}. When $G$
is infinite, we say that \emph{percolation} occurs if there exist an infinite
connected component in the subgraph of $G$ induced on all open vertices (we
consider only countable graphs with bounded maximum degree). Such an infinite
connected component is referred to as an\emph{ infinite cluster}.
In general, whenever considering the probability of a graph property occurring
we shall always (unless stated otherwise) be concerned with properties of
the subgraph of $G$ induced by the set of \emph{open} vertices.
For infinite $G$, we denote by $G_p$ the random graph obtained by independent
site percolation on $G$ with parameter $p$. Let $p_c(G):=\inf \{p: \text{
percolation occurs with probability 1 in } G_p \}$. A simple
application of Kolmogorov's zero-one law implies that $p_c(G)=\inf \{p:
G_p  \text{
has an infinite cluster with positive probability} \}$.

Finally, recall the definition of the graph
${\Z}_{*}^2$. We say that $A \subset
{\Z}^{2}$ is $*$-connected, if its induced graph in ${\Z}_*^2$ is connected.
We call a cycle (path) in the graph $\Z_*^2$ a $*$-cycle ($*$-path, respectively).
Let $H$ be a vertex induced subgraph of  $\Z_*^2$. We call the connected
components of $H$ $*$-connected components.

We first prove that in the layers model the occurrence of an infinite cluster
is a 0-1 event. Recall that for any collection of random variables, $(Y_i : i \in I)$ the $\sigma$-algebra generated by them (i.e$.$ the minimal $\sigma$-algebra with respect to which they are all measurable) is denoted by $\sigma(Y_i
: i \in I)$. Let $(I_j)_{j \in \N}$ be a collection of subsets of $I$ such that for each $j$, $I \setminus I_j$ is finite and $\bigcap_{j \in \N} I_j $ is the empty set. Then the tail $\sigma$-algebra of $(Y_i
: i \in I)$ equals to $\bigcap_{j \in \N} \sigma(Y_i
: i \in I_{j})$. Loosely speaking, an event belongs to this tail $\sigma$-algebra, if for any $j$ the occurrence of the event can be determined by knowing only the value of $(Y_i
: i \in I_j)$. Alternatively, this is the case if changing the value of finitely many of the $Y_i$'s cannot effect the occurrence of the event.
\begin{lemma}
\label{lem: 0-1law1}
Let $G$ be an infinite connected graph with a countable vertex set $V$. Assume that all the degrees are finite. Let $k \in {\N}$. Then the probability
of the event that $T_k(G)$ contains an infinite cluster is either $0$ or $1$.
\end{lemma}
\begin{proof}
Let $Y_v$ be the indicator of the event that $v \in T_k(G)$. One can readily
verify that the event that $T_k$ contains an infinite cluster,
denoted by $A_k$, is in the tail $\sigma$-algebra of $(Y_v : v\in V)$.
Pick an arbitrary $u \in V$. The previous tail $\sigma$-algebra can be written
as $\bigcap_{r=1}^{\infty} \sigma(Y_v : d(v,u)>r) \subset \bigcap_{r=1}^{\infty}
\sigma(X_v : d(v,u)>r-1)$. The last inclusion is true since for any vertex $v$, the layer to which $v$ belongs to (and thus also $Y_v$) can be determined by the ages of $v$ and its neighbors. Now, $\bigcap_{r=1}^{\infty}
\sigma(X_v : d(v,u)>r-1)$ is the tail $\sigma$-algebra of a sequence of independent
random variables, hence by Kolmogorov's 0-1 law every event in  $\bigcap_{r=1}^{\infty}
\sigma(X_v : d(v,u)>r-1)$, and hence also every event in the tail $\sigma$-algebra
of $(Y_v :
v\in V)$ is a 0-1 event. This implies that indeed
$A_k$ is a 0-1 event.
\end{proof}

We say that a graph $G=(V,E)$ is vertex transitive, if for any $u,v \in V$
there exists a bijection $\phi_{v,u}:V \to V $, such that $\phi_{v,u}(v)=u$
and $a \sim b$ iff $\phi_{v,u}(a) \sim \phi_{v,u}(b)$ for any $a,b \in V$.
Note that\ for any vertex transitive graph we have that if $(X_s : s\in V)$
are i.i.d$.$ $U[0,1]$ random variables, then if we set $X'_s=X_{\phi_{v,u}^{-1}(s)}$,
then also $(X'_s : s\in V)$ are i.i.d$.$ $U[0,1]$ random variables.

Fix $k\in \N$. Let $H$ and $H'$ be the graphs induced on the first $k$ layers
of $G$ with respect to $(X_s)_{s \in V}$ and $(X'_s)_{s \in S}$, respectively.
Clearly both $H$ and $H'$ are distributed as $T_k(G)$. Note that the connected
component of $v$ is infinite in $H$ iff the connected component of $u$ is
infinite in $H'$. This implies that for any $u,v \in V$, the probability
that they belong to an infinite cluster is the same.

\begin{lemma}
\label{lem:zero_one}
Let $G$ be an infinite connected graph and let $\Pr$ be a probability measure
corresponding to some percolation process on $G$. Assume that the probability
that there exists an infinite cluster is either $0$ or $1$.
Suppose that for every $v \in V$, $\Pr(|C(v)|=\infty)=\Theta$. Then $\Theta>0$
iff the probability that there exists an infinite open cluster is $1$. In
particular, in $T_k(G)$ we have that $\Theta>0$ implies that with probability
1 there exists an infinite cluster in $T_k(G)$.
\end{lemma}
\begin{proof}
If $\Theta=0$, then $$\Pr(\text{there exists an infinite cluster}) \le \sum_{v
\in V}\Pr(|C(v)|=\infty)=0.$$ If  $\Theta>0$, then pick an arbitrary $v \in
V$. $$\Pr(\text{there exists an infinite cluster}) \ge  \Pr(|C(v)|=\infty)
>0.$$ So by the zero-one assumption $\Pr(\text{there exists an infinite cluster})=1$.
\end{proof}

\begin{definition}
\label{def: cyclesurronds}
Let $C_{*}$ be a simple $*$-cycle in $\Z_*^2$. We call the finite connected
component (with respect to $\Z^2$) of $\Z^2 \setminus C_{*}$ the interior
of $C_{*}$. Similarly, for a simple cycle $C$ in $\Z^2$, we call the finite
$*$-connected component (with respect to $\Z_*^2$) of $\Z^2 \setminus C$
the interior of $C$. Let $A \subset \Z^2$. We say that a (simple) cycle or
$*$-cycle $C$ surrounds $A$, if $A$ is contained
in the union of $C$ and its interior.
\end{definition}

Throughout we  consider only simple cycles ($*$-cycles) even when not mentioned
explicitly.

A basic topological tool in percolation theory is the fact that in ${\Z}^2$,
$(0,0)$ does not belong to an infinite cluster of open vertices iff there
is a simple cycle in ${\Z}_{*}^2$ around $(0,0)$ consisting only of closed
vertices. The following lemma, whose proof is omitted, generalizes this principle.

\begin{lemma}
\label{lem:square}
Suppose we partition the vertices of $G={\Z}^2$ to open and closed vertices
and call the induced graphs with respect to $\Z^2$ and $\Z_*^2$ on the set
of open vertices $H$ and $H_*$, respectively.
Let $A$ be a connected set in ${\Z}^2$. Then, there exists $v \in A$ contained
in an infinite cluster of $H$ iff there does not exist a $*$-cycle in ${\Z}_{*}^2$
consisting
of closed vertices surrounding $A$. Similarly, a $*$-connected set $A$ in
${\Z}_{*}^2$ will contain a vertex that belongs to an infinite $*$-connected
component of $H_*$ iff there does not exist a cycle composed of closed vertices
in $\Z^2$ surrounding $A$
\end{lemma}
Note that we allow the enclosing cycle to intersect with the internal boundary
of $A$. This is crucial as had we required the enclosing cycle to be disjoint
from the whole of $A$ the lemma would clearly be  false in the case that
all vertices of $A$ are closed. Moreover, note that in order to apply this
lemma we do not need the percolation process to be independent.

The following lemma is a classical result in percolation theory due to Russo
\cite{Russo}.
\begin{lemma}
\label{lem:russo}
In ordinary Bernoulli percolation,
$p_c({\Z}_{*}^2)+p_c({\Z}^2)=1.$
\end{lemma}
Higuchi \cite{Higuchi}  was the first to show that $p_c({\Z}^2)>1/2$
and in \cite{Berg} it was shown that  $p_c({\Z}^2)>0.556$.

Using the aforementioned lemmas we can prove the following useful statement.
\begin{lemma}
\label{lem:square_lim}
Consider independent site percolation on ${\Z}^2$ with parameter $p:=1/2$ and denote the corresponding probability measure by $\Pr_{p}$. Let $(U(r) : r \in \N)$ be
a collection of connected sets such that $U(r) \subset U(r+1)$ for all $r
\in \N$ and $ \bigcup_{r \in \N}U(r)=\Z^2$.
Define $A(r)$ to be the event
that there exists a cycle $C$ in ${\Z}^2$
composed of open vertices which surrounds $U(r)$. Then, $\lim \limits_{r
\rightarrow \infty}\Pr[A(r)]=0$.

\end{lemma}
\begin{proof}
 Let $H$ be the vertex induced graph on the set of closed vertices w.r.t. $\Z_{*}^2$. By Lemma \ref{lem:square} $D(r):=(A(r))^c$ is the event
that there exists a vertex in $U(r)$ which is contained in an infinite connected component of $H$.
Since $(U(r): r \in \N)$, is an increasing collection of sets that exhausts
$\Z^2$, the event that there exists an infinite connected component of $H$  is the increasing
limit of the events $D(r)$. By our assumption that $p< p_{c}^{site}(\Z^2)$ in conjunction with Lemma \ref{lem:russo}, we have that $\Pr[H \text{ has an infinite connected component}]=1$. Hence $\lim_{r\rightarrow \infty} \Pr[A(r)]=1-\lim_{r\rightarrow \infty}
\Pr[D(r)]=0$.
\end{proof}
\subsection{Proof of Theorem 1.4}

Recall the definitions of $V_{even}$, $V_{odd}$, $\Z_{even}^{2}$, $\Z_{odd}^{2}$
and $\Z_*^2$.
\begin{lemma}
\label{lem: ZevenisomorphictoZ2}
Both ${\Z}_{even}^2$ and ${\Z}_{odd}^2$ are isomorphic to ${\Z}^2$.
\end{lemma}
\begin{proof}
The transformation $\phi(u,v)=(u+v,u-v)$ provides an isomorphism between
${\Z}^2$ to ${\Z}_{even}^2$. ${\Z}_{even}^2$ is
clearly isomorphic to ${\Z}_{odd}^2$  by
the translation $v \to v+(1,0)$.
\end{proof}

We call a $*$-cycle contained in $V_{even}$ (respectively, $V_{odd}$) an
\emph{even-cycle} (respectively, \emph{odd-cycle}). Obviously, an even-cycle
(odd-cycle) is just a cycle in $\Z_{even}^2$ (respectively, $\Z_{odd}^2$).
We say that an even/odd-cycle surrounds $A \subset \Z^2$, if thought of as
a $*$-cycle in $\Z_*^2$, it surrounds $A$ in terms of Definition \ref{def:
cyclesurronds}.

\begin{lemma}
\label{lem: cyclesondifferentparity}
Let $A \subset \Z^2$ be a connected set of vertices. Then, there exists some
$a \in A$ which belongs to an infinite cluster of $T_4(\Z^2)$ iff there does
not exist an even-cycle surrounding $A$
 composed of vertices in $L_5(\Z^2)\cap V_{even}$ and there does not exist
an odd-cycle surrounding $A$
  composed of vertices in $L_5(\Z^2)\cap V_{odd}$.
  \end{lemma}
\begin{proof}
By Lemma \ref{lem:square} there does not exist a vertex $a \in A$ that belongs
to an infinite cluster of $T_4(\Z^2)$ iff there exists a $*$-cycle consisting
of vertices in $L_5(\Z^2)$ which surrounds $A$. Note that since $L_5$ is
an independent set, such a $*$-cycle must be contained in either $V_{even}$
or in $V_{odd}$.
\end{proof}
We can now prove the existence part of Theorem \ref{thm:grid}. The argument
of the proof
is simple. Since independent site percolation on $\Z^2$ with parameter 1/2
is subcritical, then for a large connected set in $\Z^2$  independent site
percolation with parameter 1/2 will contain a cycle composed of closed vertices
contained in $V_{even}$ (the same holds for $V_{odd}$) surrounding it only
with some positive probability which can be made arbitrary small by picking
an arbitrary large set. By Lemma \ref{lem:order} the same holds for $L_5(\Z^2)$.

\begin{theorem}
\label{thm: T_4(Z2}
For ${\Z}^2$, $T_4$ contains an infinite cluster with probability 1.
\end{theorem}
\begin{proof}
Consider $V_{even}\cap L_5(\Z^2)$ as a percolation process on $\Z_{even}^2$.
By Lemma~\ref{lem:order} the aforementioned percolation process is stochastically
dominated by independent site percolation on ${\Z}_{even}^{2}$ with parameter
$1/2$.

   Let $A_r(even)$ be the event
that there is an even-cycle consisting of vertices belonging to
${\Z}_{even}^2 \cap L_5(\Z^2)$ which surrounds $[-r,r]^2 \cap V_{even}$.
Define $A_r(odd)$ in an analogous manner with respect to ${\Z}_{odd}^2$.
By Lemmas \ref{lem:square_lim}  (here applied on $\Z_{even}^2$
instead of on $\Z^2$) and \ref{lem: ZevenisomorphictoZ2} we know that $\Pr[
A_r(even)] \to 0$  as $r$ tends to infinity. Similarly, $\lim_{r \rightarrow
\infty}\Pr[ A_r(odd)]=0$. Fix $r$ sufficiently large such that $\Pr[ A_r(even)]<\frac{1}{10}$
and $\Pr[ A_r(odd)]<\frac{1}{10}$. Therefore the probability of the event
$A_r(even) \cup A_r(odd)$ is smaller than 1 and the result now follows from
Lemmas \ref{lem: cyclesondifferentparity} and \ref{lem:zero_one}.
\end{proof}
In the remaining of this section we establish the uniqueness part of Theorem
\ref{thm:grid} and establish Theorem \ref{thm: scalling}, which is the finite
analog of Theorem 1.4 concerning $T_4$ considered on finite boxes of the
form $[-n,n]^2 \cap \Z^2$.

The following lemma is a particular case of a fundamental result in percolation
theory about the exponential decay of the cluster size distribution in subcritical
independent site percolation due to Menshikov \cite{Mensh} and independently
Aizenman and Barsky \cite{Aizenman}.
\begin{lemma}
\label{lem: expdecayandcrossing}
Consider independent site percolation on $\Z^2$ with $p=1/2$. Denote the
corresponding probability measure by $\Pr_{1/2}$. For any $v \in \Z^2$ let
$C(v)$ be the open cluster of $v$ (i.e$.$ the connected component of $v$
in the graph induced on the set of open vertices, where if $v$ is closed,
we define $C(v)$ to be the empty set). Then there exists a constant $M>0$
such
that for any $v \in \Z^2$, $\Pr_{1/2}[|C(v)| \ge k] \le e^{-Mk }$.
\end{lemma}
\begin{definition}
\label{def: crossing}
Let $v=(v_1,v_2),u=(u_{1},u_2) \in \Z^2$ such that $v_i < u_i$, $i=1,2$.
Consider the rectangle $I=I_{v,u}:=\{(w_1,w_2)
\in \Z^2 : v_1 \le w_1 \le u_{1}, v_2 \le w_2 \le u_2 \}$. Denote $L:=\{(w_1,w_2)\in
I: w_1=v_1 \}$, $R:=\{(w_1,w_2)\in
I: w_1=u_1 \}$,  $D:=\{(w_1,w_2)\in
I: w_2=v_2 \}$ and $U:=\{(w_1,w_2)\in
I: w_2=u_2 \}$. We call a path in $\Z^2$ from $L$
to $R$ (from $D$ to $U$) which is contained in $I$
a $LR$ crossing ($DU$ crossing, respectively). We call a $*$-path from $L$
to $R$ (from $D$ to $U$) in $\Z_*^2$
contained in $I$ a
$LR$ $*$-crossing ($DU$ $*$-crossing, respectively).
\end{definition}
\begin{lemma}
\label{lem: crossing}
Suppose we partition the vertices of a rectangle $I$ to open and closed vertices
and call the induced graphs with respect to $\Z^2$ on the set
of open vertices $O$. Call the induced graph with respect to $\Z_*^2$ on
the set of closed vertices $F$.
Then either there exist a $LR$ crossing in $O$ or there exists a $DU$ $*$-crossing
in $F$. The same holds when the roles of $LR$ and $DU$ are replaced.
\end{lemma}
We omit the proof of the previous lemma.
\begin{lemma}
Let $I:=I_{(v_{1},v_{2}),(v_1+m,v_2+k)}$ be as in Definition \ref{def: crossing}
for some $m,k \in \N$. Let $J$ be the induced graph on $T_4(\Z^{2})\cap I$
with respect to $\Z^{2}$. Then,
\begin{equation}
\label{eq: crossing}
P(\exists LR \text{ crossing in } J) \ge 1- m e^{-Mk}, P(\exists DU \text{
crossing in } J) \ge 1- ke^{-Mm}.
\end{equation}
\end{lemma}
\begin{proof}
By Lemma \ref{lem: crossing} if there does not exist a $LR$ crossing in $J$,
then there exists a $DU$ $*$-crossing in $L_5(\Z^2) \cap I$. Since $L_5$
is an independent set, such a crossing is contained in either $V_{even}$
or $V_{odd}$. Pick $d \in D \cap V_{even}$. We argue that the probability
that there exists a $DU$ $*$-crossing starting from $d$ in $L_5(\Z^2)\cap
\Z_{even}^2$
is at most $e^{-Mk}$, where $M$ is as in Lemma \ref{lem: expdecayandcrossing}.
The same holds for any $d \in D \cap V_{odd}$.  To see this, note that the
graph distance of $d$ from $U$ with respect to $\Z_{even}^2$ is $k$. The
event that there exists a $DU$ $*$-crossing starting from $d$   is clearly
contained in the event that the
size of the connected component of $d$ in the induced subgraph  on $L_5(\Z^2)\cap
V_{even}$ with respect to $\Z_{even}^2$ is of size at least $k$. We can upper
bound the last probability by Lemmas \ref{lem: ZevenisomorphictoZ2},
\ref{lem:order} and \ref{lem: expdecayandcrossing}.
It follows by a union bound over the vertices of $D$ that the probability
that there exists a $DU$
$*$-crossing in $L_5(\Z^2)$ is at most $me^{-Mk}$.  The second inequality
is proven in an analogous manner.
\end{proof}
\begin{theorem}
\label{thm: uniquenessgrid}
$T_4(\Z^2)$ contains a unique infinite cluster $a.s.$
\end{theorem}
\begin{proof}
For any $k \in \N$ in the notation of Definition \ref{def: crossing} let:
$I_1(k):=I_{(-2^{k+1},-2^{k+1}),(-2^{k},2^{k+1})}$, $I_2(k):= I_{(-2^{k+1},-2^{k+1}),(2^{k+1},-2^{k})}$,
$I_3(k):=  I_{(-2^{k+1},2^{k}),(2^{k+1},2^{k+1})}$ and $I_4(k):=  I_{(2^{k},-2^{k+1}),(2^{k+1},2^{k+1})}$.
For any $i \in [4]$ and $k \in \N$ we let $J_i(k)$ be the induced
graph with respect to $\Z^{2}$ on $T_4(\Z^{2})\cap I_{i}(k)$.

We say that $I_i(k)$ is \emph{good} if $J_i(k)$ contains a $LR$ and a $UD$
crossing. Let $G_k$ be the event that $I_i(k)$ is \emph{good} for all $1
\le i \le 4$. Denote the complement of the event $G_k$ by $G_k^c$. By  (\ref{eq:
crossing}) and a union bound, for any $k \in \N$  $$\Pr[G_k^{c}] \le 8
\cdot 2^{k+1}e^{-M2^{k}} \le Ce^{-M'2^{k}} \text{ for some $0<C$ and }0<M'<M.$$
Hence $\sum_{k}\Pr[G_k^{c}] < \infty$. It follows
from the Borel-Cantelli Lemma that $a.s.$ all but finitely many of the events
$(G_k : k \in \N)$ occur. Note that if $G_k$ occurs, then there must be a
 cycle composed of vertices in $T_4(\Z^2)$ which is contained in the annulus
$\bigcup_{i=1}^4 I_i(k) $ that surrounds the rectangle $I_{(-2^{k},-2^{k}),(2^{k},2^{k})}
$ (which is the interior of that annulus).
Namely, $G_k$ implies that we have a $LR$ crossing of $I_2(k)$ and of $I_3(k)$
in $T_4$ and a $DU$ crossing of  $I_1(k)$ and of $I_4(k)$
in $T_4$. The union of which contains the desired cycle.

Pick such a cycle for each $k$ for which $G_k$ holds and call it $C_k$. Let
$u,v \in \Z^2$. With probability 1 there exists some $k$ sufficiently large
such that both $u$ and $v$ are contained in the interior of the annulus
$\bigcup_{i=1}^4
I_i(k) $ and $G_k$ holds. If both $|C(v)|,|C(u)|=\infty$, then both of $C(u)$ and $C(v)$ must intersect $C_k$ (where $C(u)$ and $C(v)$ are the components of $v$ and $u$, respectively, in $T_4(\Z^2)$). So with probability 1 $C(v)=C(u)$.
\end{proof}
We now comment about how the uniqueness proof also provides an alternative
proof for the existence of the infinite cluster of $T_4(\Z^2)$.
Define $I_{k}'$ and $J'_{i}(k)$
in an analogous manner to the definitions of $I_i(k)$ and $J_i(k)$, where
$2^{k},2^{k+1}$ are replaced by $4^k,4^{k+1}$,
respectively ($i \in [4]$).
Define $G_k'$ with respect to the rectangles $J'_i(k)$ $i
\in [4]$ in an analogous manner to the definition on $G_k$. A similar
calculation as in the above proof shows that $a.s.$ all but finitely many
of the events $G_k'$ occur. Observe that $\bigcup_{i=1}^4
I'_i(k+1)=\bigcup_{j=0}^{1}\bigcup_{i=1}^4
I_i(2k+j)$.
Note that since we require in $G_k$ and $G_k'$ that every rectangle out of
the 4 corresponding to index $k$ would have both a $LR$ crossing and a $DU$
crossing, we get that on $G_k' \cap G_{2k} \cap G_{2k+1}$ the crossings of
$G_k'$ must connect the cycles we get in the above proof in the annuli
$\bigcup_{i=1}^4
I_i(2k) $ and $\bigcup_{i=1}^4
I_i(2k+1) $ from the occurrence of $G_{2k}$ and $G_{2k+1}$. Since with probability
one this occurs for all but finitely many $k$'s, we get that $a.s.$ $T_4(\Z^2)$
has an infinite cluster.

The proof  of the next theorem essentially follows that of Theorem 7.61 in
Grimmett \cite{Grimmett}.
\begin{theorem}
\label{thm: scalling}
Let $G_{n}:=(V_{n},E_{n})$ with $V_{n}:=[-n,n]^2 \cap \Z^2$, be the induced
graph on $V_{n}$ with respect to $\Z^2$. Let $\Theta$ be the probability
that $0$ belongs to an infinite cluster in $T_4(\Z^2)$. Take a fixed $\epsilon \in (0,1)$.
Let $GC$ be the largest connected component of $T_4(G_{n})$.
\begin{itemize}
\item[(i)] There exist positive absolute constants $C,M$ such that for any
$n$ sufficiently large
$$ \Pr \left[|GC|<4n^{2}(1-\epsilon)\Theta \right] \le e^{-C \epsilon^2
n^{1/5}}+4n^{2}e^{-M
n^{1/5}}.$$
\item[(ii)]
There exists an absolute constant $L$ such that with probability at
least $1-n^{-2}$ all other components of $T_4(G_{n})$ apart from $GC$ are
(as subsets of $\Z^2$) of diameter at most $L\log n$.
\end{itemize}
\end{theorem}
\begin{proof}
We may assume that $T_4(G_{n})$ is obtained from sampling $(X_v : v \in \Z^2)$,
where as usual these are i.i.d$.$ $U[0,1]$ random variables. We may use $(X_v)_{v
\in \Z^2}$ to define simultaneously $T_4(G_{n})$ for all $n \in \N$ together
with $T_4(\Z^2)$.
Note that $T_4(G_n)$ and $T_4(\Z^2)$ agree on $[-n+1,n-1]^2 \cap \Z^2$. Let
$A_n:=[-n+ 4\lceil n^{1/5} \rceil , n - 4\lceil n^{1/5} \rceil]^2 \cap V_{n}$.
For every $v \in A_n$, let $C(v)$ be the connected component of $v$ in $T_4(G_n)$.
For $v \in A_n$, let $Z_v$ be the indicator of the event that as a subset
of $\Z^2$ the diameter of $C(v)$ is at least $4n^{1/5}$. Note that the last
event is a subset of the event that $v$ belongs to the infinite cluster of
$T_4(\Z^2)$ (this is meaningful by our coupling of $T_4(G_n)$ with $T_4(\Z^2)$)
and hence has probability at least $\Theta$. Note that each $Z_v$ depends
on at most $64n^{2/5}$  random variables from $(Z_u)_{u \in A_n}$. Let $B$ be the event that $\sum_{v \in A_n}Z_v
< 4n^{2}(1-\epsilon)\Theta$. By Azuma
inequality, for every $n$ sufficiently large
$$\Pr \left[B \right] \le e^{-C \epsilon^2 n^{1/5}},$$ for some
positive absolute constant $C$.

In the notation of Definition \ref{def: crossing}, look at all the rectangles
of the form $I_{(a,-n),(a+\lceil n^{1/5} \rceil,n) \rceil)} $ and of the
form $I_{(-n,a),(n,a+\lceil n^{1/5} \rceil)} $ for any integer $-n \le a
\le n- \lceil n^{1/5} \rceil$. Then by Lemma \ref{lem: crossing} and taking a union
bound over all such rectangles, we have that with probability at least $1-4n^{2}e^{-M
n^{1/5}}$ for all  $-n \le a
\le n- \lceil n^{1/5} \rceil$  there is a $DU$ crossing of the rectangle
$I_{(a,-n),(a+\lceil n^{1/5} \rceil,n) \rceil)}$ contained in $T_4(G_n)$
and a $LR$ crossing of the rectangle $I_{(-n,a),(n,a+\lceil n^{1/5} \rceil)}$
contained in $T_4(G_n)$. Call this event $D$. It is easy to see that on $B
\cap D$ there is a connected component of $T_4(G_n)$ of size at least $4n^{2}(1-\epsilon)\Theta$
and every other connected component
is of diameter at most $2\lceil n^{1/5} \rceil$.
 This concludes the proof of (i). The proof of (ii) is obtained in a similar
manner by considering all the rectangles of the form $I_{(a,-n),(a+L\log
n,n) \rceil)} $ and of the
form $I_{(-n,a),(n,a+L\log n)}$ for some sufficiently large constant $L$.
We omit the details.
\end{proof}

\section*{Acknowledgements}
The authors would like to thank Itai Benjamini, Gady Kozma and Shai Vardi
for useful discussions.

\begin{appendix}

\section{An upper bound on the size of monotone components}
\label{sec:vardi}

In this Section we sketch a proof of Theorem~\ref{thm:vardi} (based on the proof given in~\cite{vardi}).
We start with the special case of trees. The following proposition is essentially from~\cite{NO}.

\begin{proposition}
\label{pro:expectation}
Let $G$ be a $d$-ary tree (every vertex has $d$ children). Let $a$ be the random label of the root $r$. Then the expected size of $r$'s monotone component (expectation taken over choice of other random labels) is $E[|C_{mon}(r)|] = e^{ad}$.
\end{proposition}

\begin{proof}
Denote the level of the root $r$ by~0, and consider a vertex $u$ at level $i$. The probability that $u \in C_{mon}(r)$ is exactly $a^i/i!$. Hence by linearity of expectation, $E[|C_{mon}(r)|] = \sum_{i\ge 0} d^ia^i/i! = e^{ad}$.
\end{proof}

The following Lemma is proved in~\cite{vardi}. For completeness, we sketch its proof.

\begin{lemma}
\label{lem:vardi}
There is some constant $b > 0$ such that for a $d$-ary tree as above and every $n$, $Pr\left[|C_{mon}(r)| \ge e^{bd}\log n\right] \le \frac{1}{n^2}$.
\end{lemma}

\begin{proof}
(Sketch.) At worse, the root $r$ has age $X_{r} = 1$. Partition the range $[0,1]$ into $3d$ classes of equal size. In $C_{mon}(r)$, consider first only those edges that join two vertices of the same class. This decomposes $C_{mon}(r)$ into subtrees, where all vertices in a subtree are of the same class. As every such subtree is generated by a subcritical branching process (the expected number of neighbors of a vertex $v$ in the same class as $v$ is $1/3$), its expected size is constant, and the probability it has size $k$ decreases exponentially with $k$. Moreover, every vertex of a given class has in expectation $1/3$ of a child in any given class below it. Using these facts it is not hard to prove (by induction, starting at class~1 which is the top class) that the number of vertices of class $i$ exceeds $O(e^{ci}\log n)$ with probability at most $1/n^2$, where $c$ is some sufficiently large constant independent of $i,d,n$.
\end{proof}

The bound in Lemma~\ref{lem:vardi} is best possible up to the choice of constant $b$, as the following example shows. Let $k$ be an integer such that $d^k \simeq \frac{1}{8}\log n$. With probability at least $1/n$, all vertices of level $i$ of the tree (for $i < k$) have a label in the range $[1 - 2^{i-k-2}, 1 - 2^{i - k - 1}]$, giving $\frac{1}{8}\log n$ leaf vertices. Thereafter, Proposition~\ref{pro:expectation} implies that the expected number of descendants per leaf is exponential in $d$.

Bounds on $C_{mon}$ for trees as in Lemma~\ref{lem:vardi} extend to every graph of bounded degree, as the following corollary shows.

\begin{corollary}
\label{cor:vardi}
The distribution of $|C_{mon}(r)|$ for the root of an infinite $d$-ary tree stochastically dominates the distribution of $|C_{mon}(v)|$ for every vertex $v$ in any graph of degree at most $d$.
\end{corollary}

\begin{proof}
Given a vertex $v$ in a graph $G$ of maximum degree $d$, develop an infinite tree $T$ of arity at most $d$ from it, where $v$ serves as the root $r$, its neighbors in $T$ are all its neighbors in $G$, and the same applies recursively to every other vertex appearing in $T$. A node of $G$ will appear in multiple places in this tree. Nevertheless, give every vertex of the tree an independent random label in the range $[0,1]$. We claim that the distribution of $|C_{mon}(r)|$ in this tree stochastically dominates the distribution of $|C_{mon}(v)|$ in the original graph. We prove this claim by exposing the labels in $G$ starting at $v$, and thereafter at each step exposing the labels of the yet unexposed neighbors of that vertex $u$ that has the highest label among the vertices of the current connected component of $v$. The crucial observation is that if a vertex at the time of its exposure cannot be joined to its parent $u$ (because it has a higher label), then it has no monotone path to $v$ (not even not through $u$). The same order of exposures is copied into $T$, where each vertex of $G$ is equated with its copy in $T$ that is reached by copying the chain of exposures from $G$ to $T$. As in $T$ a vertex has additional copies, later exposing their independent labels and joining them to the connected component if possible only increases its size.
\end{proof}

The combination of Lemma~\ref{lem:vardi} and Corollary~\ref{cor:vardi} imply that with probability at least $1 - 1/n$, no monotone component in $G$ is of size larger than $e^{bd}\log n$, proving Theorem~\ref{thm:vardi}. 


\end{appendix}

\end{document}